\documentclass[10pt]{article}
\usepackage{amsthm}
\usepackage{amsmath}
\usepackage{amssymb}
\usepackage{makeidx}

\newtheorem{definition}{Definition}[section]

\newtheorem{remark}[definition]{Remark}

\newtheorem{theorem}[definition]{Theorem}
\newtheorem{lemma}[definition]{Lemma}
\newtheorem{proposition}[definition]{Proposition}
\newtheorem{corollary}[definition]{Corollary}

\renewcommand {\l}{\lambda}
\renewcommand {\iff}{\Leftrightarrow}
\newcommand {\imp}{\Rightarrow}

\title{\Large\bf  Boolean lifting property in quantales.}
\author{\sc  Daniela Cheptea, George Georgescu }
\begin{document}
\maketitle
\vskip 1.5em
\begin{abstract}
In ring theory, the lifting idempotent property (LIP) is related to some important classes of rings: clean rings, exchange rings, local and semilocal rings, Gelfand rings,maximal rings, etc. Inspired by LIP, there were defined lifting properties for other algebraic structures: MV-algebras, BL- algebras, residuated lattices, abelian l-groups, congruence distributive universal algebras,etc.
In this paper we define a lifting property (LP) in quantales, structures that constitute a good abstraction of the lattices of ideals, filters or congruences. LP generalizes all the lifting properties existing in literature. The main tool in the study of LP in a quantale A is the reticulation of A, a bounded distributive lattice whose prime spectrum is homeomorphic to the prime spectrum os A. The principal results of the paper include a characterization theorem for quantales with LP, a structure theorem for semilocal quantales with LP and a charaterization theorem for hyperarhimedean quantales.
\end{abstract}
\section{Introduction}
 \hspace{0.5cm}In ring theory is frequently met the following lifting idempotent property (LIP): the
idempotents can be lifted modulo every left (respectively right) ideal. LIP was related to
two important classes of rings: the exchange rings and the clean rings \cite{a}. In the case of
commutative rings there was proved that the exchange rings, the clean rings and the rings
with LIP coincide. These rings have significant algebraic and topological properties and a
whole literature has been dedicated to their study (see \cite{b}, \cite{c}, \cite{d}).
By analogy to LIP, there were introduced various “lifting properties” for other algebraic
structures: MV- algebras \cite{e} , commutative l-groups \cite{f},  BL- algebras \cite{g}, residuated
lattices \cite{GeorgescuVoiculescu}, bounded distributive lattices \cite{h}, etc.
Both LIP definition and of other lifting properties assume the existence of “Boolean
centers”, subsets of algebras endowed with a Boolean structure: the idempotent in the
case of commutative rings and the complemented elements for the other mentioned
algebras.
If we want to extend the definition of these lifting properties to more general classes of
universal algebras then we must ensure that these algebras possess Boolean centers. A
suggestion can be offered by the remark that the Boolean centers of the concrete algebras
are isomorphic (or anti-isomorphic) with Boolean subalgebras of some ideal (or filter)
lattices: see Lemma 1, \cite{Banaschewski} for commutative rings,  Lemmas 5.6 and 5.7, \cite{GeorgescuVoiculescu} for
residuated lattices, etc. In order to obtain a Boolean center of an algebra $A$ we shall
choose a Boolean subalgebra of the congruence lattice $Con(A)$. In \cite{GeorgescuMuresanJMVL}, \cite{GeorgescuMuresanStudia} there are defined
two notions of lifting property for a congruence distributive algebra $A$: Congruence
Boolean Lifting Property (CBLP), whenever the Boolean center is the set $B(Con(A))$ of
complemented congruences of A, and Factor Congruence Lifting Property (FCLP),
whenever the Boolean center is the set FC(A) of factor congruences of $A$.

On the other hand, the quantales \cite{Rosenthal}, \cite{GCM} and the frames \cite{Johnstone} constitute a good abstraction
of lattices of ideals, filters and congruences. Several results in algebra, topology,
analysis, etc can formulated in the  framework of quantales and frames. The first
abstract formulation of LIP in the context of frames is the condition (3) of Lemma 4, \cite{Banaschewski}.

In this paper we shall define a notion of lifting property (LP) in setting of the coherent
quantales. This extends LIP and the other lifting properties, as well that CBLP.
The main tool in studying LP in a quantale $A$ is the reticulation $L(A)$ of $A$, a bounded
distributive lattice whose $Spec_{Id}(L(A))$ is homeomorphic to the prime spectrum $Spec(A)$
of $A$. The assignment $A \to L(A)$ allows us to transfer some algebraic and topological
results from bounded distributive lattices to quantales. In order to apply this thesis in the
study of LP we shall use a remarkable property: the Boolean center $B(A)$ of a quantale $A$
is isomorphic to the Boolean algebra $B(L(A))$ of complemented elements of $L(A)$.

Section 2 is a presentation of some definitions and identities in quantales and frames, 
the basic properties of prime spectrum $Spec(A)$ of a coherent quantale $A$ and of the
frame $R(A)$ of radical elements of $A$ \cite{Rosenthal}, \cite{Johnstone}, \cite{GCM}.

In Section 3 we define the reticulation $L(A)$ of a coherent quantale $A$ and we prove its unicity. We recall from \cite{Georgescu} the construction of the reticulation $L(A)$ and the principal algebraic and topological connections between $A$ and $L(A)$.

Section 4 studies the relationship between the Boolean center $B(A)$ of a coherent quantale
A, the Boolean center $B(R(A))$ of the frame $R(A)$ and the Boolean algebra $B(L(A))$. These
three Boolean algebras are isomorphic, which will generate a strong transfer of properties
between $A$, $R(A)$ and $L(A)$. Particularly, we obtain a characterization of hyperarhimedean
quantales as the quantales $A$ for which the prime elements of $A$ and the maximal
elements coincide. This result extends the Nachbin theorem in the lattice theory \cite{BalbesDwinger} and
the Kaplansky characterization of regular rings \cite{Johnstone}, \cite{Simmons}.

In Section 5 we define the lifting property (LP) in a quantale $A$ following a suggestion
given by Lemma 4, \cite{Banaschewski} and we prove the equivalence of LP in $A$, LP in the frame $R(A)$ and
a lifting property of ideals in $L(A)$ (Proposition 5.7).

Section 6 is concerned with the relationship between LP and the properties of normality
and B- normality (in $A, R(A)$ and $L(A)$). Firstly we prove that the coherent quantale is B-normal iff the frame $R(A)$ is B-normal iff the reticulation $L(A)$ is B-normal as lattice. The
main theorem of section (Theorem 6.6) combines this result on B-normality of the three
entities ( $A, R(A)$ and $L(A)$) with the mentioned proposition of the previous section, by
using Lemma 4, \cite{Banaschewski}.

In Section 7 we study LP versus finite products of quantales. The main result of section
(Theorem 7.9) establishes many equivalent conditions that characterize the semilocal
coherent quantales with LP. Particularly, $A$ is a semilocal quantale with LP if and only if
$A$ is isomorphic with a finite product of local quantales.
\section{Preliminaries}
 \hspace{0.5cm}In this section we shall recall some definitions and basic properties of quantales (see \cite{Rosenthal}, \cite{Paseka} ).
\begin{definition} A {\em quantale } is a structure of the form $(A,\lor, \land, \cdot, 0, 1  )$ such that  $(A,\lor, \land, 0, 1  ) $ is a complete lattice and $(A, \cdot)$ is a semigroup with the property that the multiplication  $\cdot $ satisfies the infinite distributive law: for any $a\in A $ and $(b_{i})_{i\in I}\subseteq A $,   
$a\cdot (\bigvee _{i\in I} b_{i})=\bigvee_{i\in I}(a\cdot b_{i}) $ and  $(\bigvee _{i\in I}b_{i})\cdot a=\bigvee_{i\in I}(b_{i}\cdot a) $  
\end{definition}
If the multiplication $\cdot $ coincides with $\land$ then  the notion of frame \cite{Johnstone} is obtained.
A quantale $A$ is said to be 
\begin{itemize}
\item []{\em unital}, if $(A, \cdot, 1)$ is monoid
\item []{\em commutative}, if the multiplication is commutative.
\end{itemize}

Throughout this paper by a quantale , we shall understand a unital and comutative quantale. Often we shall write $ab$ instead of $a\cdot b$.

Let us denote by $ K(A)$ the set of compact elements of a quantale $A$. An algebraic quantale $A$ is {\em coherent} if $K(A)$ is closed under $\cdot $ and $1\in K(A)$.

Let $A$ be a quantale. For any $a,b\in A$, denote $a\to b=\bigvee\{x\in A| a\cdot x\leq b\}$. Therefore $(A,\lor, \land, \cdot,\to,  0, 1  )$  is a residuated lattice \cite{Kowalski}, \cite{JipsenTsinakis}. The negation operation $a^{\perp}$ is defined by $a^{\perp}=a\to 0=\bigvee \{x\in A|a\cdot x=0\}$.

The following lemma collects some basic properties of $\to$ and $\perp$.
\begin{lemma}\label{to}
For all $a, b, c \in A$ we have  $a\leq b\to c$  iff $a\cdot b \leq c$. 
\end{lemma}
\begin{lemma}\label{propQ}
For all $a, b, c \in A$, the following properties hold:
\newcounter{nr}
\begin{list}{(\arabic{nr})}{\usecounter{nr}}
\item If $a\lor b =1$ then $a\cdot b= a \land b $;
\item If $a\lor b=a\lor c=1$ then $a\lor (b\cdot c)= a \lor (b\land c)=1$;
\item If $a\lor b =1$ then $a^{n}\lor b^{n}=1$ for any natural number $n\geq 1$;
\item If $a\lor b =1$ and $a\leq c$ then $a\lor (b\cdot c)=c$. 
\end{list}
\end{lemma}
\begin{proof}
The properties (1), (2) follow by \cite {Birkhoff} and (3) is obtain by induction.\\
(4) Assuming $a\lor b=1$ and $a\leq c$ one gets $c=(a\lor b) \cdot c= a\cdot c\lor b\cdot c \leq a\lor b\cdot c$. The converse inequality follows from $a\leq c$ and $b\cdot c \leq c$.
\end{proof}
\begin{definition}
An element $p < 1$ of a quantale $A$ is {\em m-prime} if for all $a,b\in A, a\cdot b \leq p$ implies $a\leq p$ or $b\leq p$. Let us denote by $Spec(A)$ the set of m-prime elements in $A$ and $Max(A)$ the set of maximal elements of $A$.
\end{definition}
\begin{lemma}\label{SpecA inclus MaxA}
Assume $1\in K(A)$. Therefore:
\begin{list}{(\arabic{nr})}{\usecounter{nr}}
\item For any $a < 1$ there exists $m\in Max(A)$ such that $a\leq m$;
\item $Spec(A)\subseteq Max(A)$
\end{list}
\end{lemma}
\begin{lemma}
Let  $A$ be a coherent quantale and $p\in A\setminus \{1\}$. The following are equivalent:
\begin{list}{(\arabic{nr})}{\usecounter{nr}}
\item $p$ is m-prime;
\item For all $c, d\in K(A)$, $c\cdot d\leq p$ implies $c\leq p$ or $d\leq p$;
\end{list}
\end{lemma}
For any $a\in A$ denote $V(a)=\{p\in Spec(A)|a\leq p\}$ and $v(a)= V(a)\cap Max(A)$. Then $Spec(A)$ (respective Max(A)) can be endowed with a topology whose closed sets are $(V(a))_{a\in A}$ (respective $(v(a))_{a\in A}$). 

The {\em radical} $\rho_{A}(a)$ of an element $a\in A$ is defined by $\rho_{A}(a)=\bigwedge \{ p\in Spec(A)|a\leq p\}$; if $a=\rho_{A}(a)$ then $a$ is a {\em radical element}. If there is no danger of confusion then we shall write $\rho(a)$ instead $\rho_{A}(a)$.
\begin{lemma}\cite{Rosenthal}\label{rho} For all $a,b\in A$ the following properties hold:
\begin{list}{(\arabic{nr})}{\usecounter{nr}}
\item $a \leq \rho(a) $;
\item $\rho(a\land b)=\rho(a\cdot b)=\rho(a)\land \rho(b)$;
\item $\rho(a)=1$ iff $a=1$;
\item $\rho(a\lor b)=\rho(\rho(a)\lor\rho(b))$;
\item $\rho(\rho(a))=\rho(a)$.
\item $\rho(a)\lor\rho(b)=1 $ iff $a\lor b=1$
\item $\rho(a^{n})=\rho(a)$ for all integer $n\geq 1$.
\end{list}
\end{lemma}
The equality (4) of previous Lemma can be extended to an arbitrary family $(a_{i})_{i\in I}$ of $A$:
$\rho(\bigvee_{i\in I}a_{i})=\rho(\bigvee_{i\in I}\rho(a_{i}))$.

Let us denote by $R(A)$ the set of radical elements of $A$. For any family $(a_{i})_{i\in I}\subseteq R(A)$, let us denote $\dot{\bigvee}_{i\in I}a_{i}=\rho(\bigvee_{i\in I}a_{i})$. Thus $(R(A), \dot{\bigvee}, \bigwedge, \rho(0), 1)$ is be a frame.
\begin{lemma}\label{coherent=>MaxA=MaxRA}
If $A$ is a coherent quantale then $Max(A)=Max(R(A))$ and $Spec(A)=Spec(R(A))$.
\end{lemma}
\begin{proof}
Assume that $m\in Max(A)$ then $\rho(m)=m$, hence $Max(A)\subseteq Max(R(A))$. Conversely, assume that $m\in Max(R(A))$ and $a\in A$ such that $m\leq a<1$. Thus, by Lemma \ref{SpecA inclus MaxA}, one gets $m\leq a\leq \rho(a)< 1$, so $m=\rho(a)=a$. It follows that $Max(R(A))\subseteq Max(A)$.  
\end{proof}

\begin{lemma}\cite{Martinez}\label{c^{k}}
Let $A$ be a coherent quantale and $a\in A$. Then 
 \begin{list}{(\arabic{nr})}{\usecounter{nr}}
\item $\rho(a)=\bigvee\{c\in K(A)|c^{k}\leq a \, , \, for \, some\, k\geq 1 \}$;
\item For any $c\in K(A)$, $c\leq \rho(a)$ iff $c^{k} \leq a$, for some $k\geq 1$.
\end{list}
\end{lemma}
\begin{lemma}
If $A$ is a coherent quantale then $K(R(A))=\rho(K(A))$ and $R(A)$ is a coherent frame. 
\end{lemma}
\begin{proof}
Let $a\in K(A)$ and $(x_{i})_{i\in I}\subseteq A$ such that $\rho(a)\leq \dot{\bigvee}_{i\in I}\rho(x_{i})$. Then $a\leq \rho(\bigvee_{i\in I}\rho(x_{i}))$, hence by Lemma \ref{c^{k}}, there exists a natural number  $k\geq 1$ such that $a^{k}\leq \bigvee_{i\in I} \rho(x_{i})$. Since $a^{k}\in K(A)$, there exists a finite subset $J$ of $I$ such that $a^{k}\leq \bigvee_{i\in J}\rho(x_{i})$, therefore $\rho(a)=\rho(a^{k})\leq \rho(\bigvee_{i\in J}\rho(x_{i}))=\dot{\bigvee}_{i\in J}\rho(x_{i})$. It follows that $\rho(a)\in K(R(A))$, so $\rho(K(A))\subseteq K(R(A))$.
In order to prove the converse inclusion, assume that $\rho(x)\in K(R(A))$. We can write $\rho(x)=\bigvee_{i\in I}c_{i}$, for some family $(c_{i})_{i\in I}$ of compact elements of $A$, hence $\rho(x)=\rho(\rho(x))=\rho(\bigvee_{i\in I}c_{i})=\dot{\bigvee}_{i\in I}\rho(c_{i})$. By hypothesis $\rho(x)$ is a compact element of $R(A)$, so there exists a finite subset $J$ of $I$ such that $\rho(x)=\dot{\bigvee}_{i\in J}\rho(c_{i})=\rho(\bigvee_{i\in J}c_{i})$. $\bigvee_{i\in J}c_{i}$ is a compact element of $A$, hence $\rho(x)\in \rho(K(A))$. We conclude that $K(R(A))\subseteq \rho(K(A))$. By Lemma \ref{rho},     $ \rho(K(A))$ is closed under $\land $, so $R(A)$ is a coherent frame.
\end{proof}  
The quantale $A$ is {\em semiprime} if $\rho_{A}(0)=0$.

Let $A,B$ be two quantales. A function $f:A \to B$ is a morphism of quantales if it preserves the arbitrary joins and the multiplication; $f$ is an unital morphism if $f(1)=1$. If $f(K(A))\subseteq K(B)$ then we say that $f$ preserves the compacts. 

If $L$ is a bounded distributive lattice then $Id(L)$ (respective $Spec_{Id}(L)$) will denote the set of ideals (respective prime ideals) of $L$. 
The set of maximal ideals of $L$ will be denoted by $Max_{Id}(L)$. The sets $Spec_{Id}(L)$ and $Max_{Id}(L)$ are topological spaces with respect to the Stone topologies.
For any $I\in Id(L)$, $p_{I}:L\to L/I$ will be the canonical lattice morphism defined by $p_{I}(a)=a/I$, for any $a\in L$.

A bounded distributive lattice $L$ is {\em Id-local} if $|Max_{id}(L)|=1$. $L$ is {\em Id-semilocal} if it has a finite number of maximal ideals. It is well-known that $L$ is Id-local iff for all $a,b\in L$, $a\lor b$ implies $a=1$ or $b=1$.
 
The {\em boolean center} of the lattice $L$ is the Boolean algebra $B(L)$ of the complemented elements in $L$.
If $I\in Id(L)$ then $p_{I}:L\to L/I$ induces a boolean morphism $B(p_{I}): B(L)\to B(L/I)$. Following \cite {h}, we say that the lattice $L$ has {\em Id-BLP} if for all ideals $I$ of $L$, the boolean morphism $B(p_{I})$ is surjective.

Recall from \cite{Johnstone} that the bounded distributive lattice $L$ is {\em normal} if for all $a,b\in L$, with $a\lor b= 1$, there exist $c,d\in L$ such that $a\lor c=b\lor d =1 $ and $c\cdot d= 0$.
 
\section{Reticulation of a coherent quantale}
 \hspace{0.5cm}In this section we present an axiomatic definition of the reticulation of a coherent quantale, we prove its unicity and we recall from \cite{Georgescu} a construction of this object.

Let $A$ be a coherent quantale and $K(A)$ the set of its compact elements. 

\begin{definition}
A {\em reticulation} of a quantale $A$ is a bounded distributive lattice $L$ together a surjective function $\l:K(A)\to L$ such that for all $a,b \in K(A)$ the following properties hold:
\newcounter{nrr}
\begin{list}{(\roman{nrr})}{\usecounter{nrr}}
\item $\l(a\lor b)\leq \l(a)\lor \l(b)$
\item $\l(a\cdot b)=\l(a)\land \l(b)$
\item $\l(a)\leq \l(b)$ iff $a^{n}\leq b$, for some integer $n\geq 1$
\end{list}
\end{definition}
 
 The reticulation $A$ will be denoted by $(L,\l)$ ; often we shall say that the lattice $L$ is a reticulation of $A$.
\begin{lemma}\label{l(compact)}
Assume that $(L,\l)$ is a reticulation of $A$. For all $a,b\in K(A)$, the following properties hold:
\begin{list}{(\arabic{nr})}{\usecounter{nr}}
\item $a\leq b$ implies $\l(a)\leq \l(b)$;
\item $\l(a\lor b)=\l(a)\lor \l(b)$;
\item $\l(a)=1$ iff $a=1$;
\item $\l(0)=0$;
\item $\l(a)=0$ iff  $a^{k}=0$, for some integer $k\geq 1$;
\item $\l(a^{n})=\l(a)$, for any $n\geq 1$. 
\end{list}
\end{lemma} 
\begin{proof}
$(1)$ follows by  $(iii)$  and $(2)$ by $(i)$ and  $(1)$.\\
$(3)$ According the surjectivity of $\l$, there exists $a\in A$ such that $1=\l(a)$. By $(1)$, $a\leq 1$ implies $1=\l(1)\leq(1)$ so $\l(1)=1$. Conversely, assume that $\l(a)=1$. Thus $\l(a)=\l(1)$, hence, by $(iii)$ there exists an integer $n\geq 1$ such that $1=1^{n}\leq a$, so $a=1$.\\
$(4)$ By the surjectivity of $\l$, there exists $a\in K(A)$ such that $0=\l(a)$. Thus, by $(1)$, $0\leq x $ impies $\l(0)\leq \l(x)=0 $, so $\l(0)=0$.\\
$(5)$ Assume $\l(a)=0$, so by $(4)$ we have $\l(a)\leq \l(0)$. According to $(iii)$ there exists an integer $n\geq 1$ such that $a^{n}\leq 0$, hence $a^{n}=0$.\\
$(6)$ By the axiom $(iii)$.     
\end{proof} 
\begin{proposition}\label{unicRetic}(the unicity of reticulation). If $(L,\l:A\to L)$ and $(L',\l':A\to L')$ are two reticulations of the quantale $A$ then there exists an isomorphism of bounded distributive  lattices $f:L\to L'$ such that the following diagram is commutative:
\begin{center}
\begin{picture}(100,70)
\put(-10,50){$K(A)$}
\put(75,50){$L$}
\put(75,0){$L'$}
\put(15,55){\vector(1,0){60}}
\put(10,45){\vector(2,-1){65}}
\put(80,45){\vector(0,-1){30}}
\put(40,60){$\l$}
\put(25,20){$\l'$}
\put(85,30){$f$}
\end{picture}
\end{center}
\end{proposition}
\begin{proof}Assume that $a,b\in K(A)$ and $\l(a)\leq \l(b)$, hence , by $(iii)$ there exists an integer $n\geq 1$ such that $a^{n}\leq b$. Thus by Lemma \ref{l(compact)},(1) and (6), one gets $\l'(a)=\l'(a^{n})\leq \l'(b)$. Therefore, for all $a,b \in K(A)$ the following equivalence holds: $\l(a)=\l(b)$ iff $\l'(a)=\l'(b)$.  \\
Therefore one can define two functions $f:L\to L'$, $g:L'\to L$ such that $f(\l(a))=\l'(a)$ and $g(\l'(a))=\l(a)$, for all $a\in A$. By using definition and Lemma \ref{l(compact)}, one can prove that $f$ and $g$ are morphisms of bounded distributive lattices. It is easy to see that $g\circ f=1_{L}$ and $f\circ g=1_{L'}$.
\end{proof}

Let us consider on $K(A)$ the following equivalence relation: $c\equiv d$ iff $\rho(c)=\rho(d)$. 
\begin{lemma}
Let $c, c', d, d'\in K(A)$. If $c\equiv c'$, $d\equiv d'$ then $c\lor d\equiv c'\lor d'$,  $c\cdot d\equiv c'\cdot d'$.
\end{lemma} 
Consider the quotient set $L(A)=K(A)/\equiv$ ; for any $c\in K(A)$ denote by $\hat{c}$ the equivalence class of $c$. Then $L(A)$ becomes a bounded distributive lattice with respect to the operations $\hat{c}\lor \hat{d}=\widehat{c\lor d}$,  $\hat{c}\land \hat{d}=\widehat{c\cdot d}$ and the constants $\hat{0}$, $\hat{1}$. $L(A)$ will be called the reticulation of the quantale $A$.

One defines a function $\l_{A}:K(A)\to L(A)$ by $\l_{A}(c)=\hat{c}$, for any $c\in K(A)$. Often we shall write $\l$ instead of $\l_{A}$.

\begin{proposition}\cite{Georgescu}(the existence of reticulation). The pair $(L(A),\l_{A})$ is the reticulation of the quantale $A$.
\end{proposition}

Let consider the functions $()^{*}:A \to Id(L(A))$ and $()_{*}:Id(L(A))\to A$ defined by the assignments:\\
$a\in A \mapsto a^{*}=\{\l(c)| c\in K(A), c\leq a\} \in Id(L(A))$\\
$I\in Id(L(A)) \mapsto I_{*}=\bigvee \{c\in K(A)|\l(c)\in I \}\in A$.

It is easy to see that these functions are order preserving.
\begin{lemma}\label{_^*}
\begin{list}{(\arabic{nr})}{\usecounter{nr}}
\item If $a\in A$ then $a\leq (a^{*})_{*}$;
\item If $I\in Id(L(A))$ then $(I_{*})^{*}=I$;
\item If $p\in Spec(A)$ then $(p^{*})_{*}=p$ and $p^{*}\in Spec_{Id}(L(A))$;
\item If $P\in Spec_{Id}(L(A))$ then $P_{*}\in Spec(A)$
\end{list}
\end{lemma}
According to this lemma, one can consider the functions $u:Spec(A)\to Spec_{Id}(L(A)) $ , $ v:Spec_{Id}(L(A))\to Spec(A) $ ,  defined by $u(p)=p^{*}$ and $v(P)=P_{*}$, for all $p\in Spec(A)$ and $P\in Spec_{Id}(L(A))$.
\begin{proposition}The functions $u$ and $v$ are homeomorphisms, inverse to one another.
\end{proposition}
By using the previous proposition, it follows that the restrictions of u (respective v) to $Max(A)$ (resp $Max(L(A))$) are homeomorphisms, inverse to one another.

\begin{proposition}If $a\in A$ then $\rho(a)=(a^{*})_{*}$.
\end{proposition} 
\begin{corollary}
If $a\in A$ and $I\in Id(L(A))$ then $a^{*}=(\rho(a))^{*} $ and $\rho(I_{*})=I_{*} $.
\end{corollary}
\begin{corollary}\label{Phi,PsiIsomorphism}
The functions $\Phi : R(A) \to Id(L(A)) $ , $\Psi : Id(L(A)) \to R(A) $ defined by $\Phi (a)=a^{*}$ and $\Psi (I)=I_{*} $ for $a\in R(A)$ and $I\in Id(L(A)) $ are frame isomorphisms, inverse to one another .
\end{corollary}
\begin{remark}
By using Lemma \ref{_^*}.(2) and Corollary \ref{Phi,PsiIsomorphism}, we can prove that $()_{*}$ is the left adjoint of $()^{*}$, i.e. for all $a\in A$ and $I\in Id(L(A))$, we have $I_{*}\leq a$ iff $I\subseteq a^{*}$. It is well known that $\rho :A\to R(A)$ is the left adjoint of the inclusion $i:R(A)\to A$. The following diagram is commutative:
\begin{center}
\begin{picture}(150,70)
\put(0,50){A}
\put(70,50){Id(L(A))}
\put(10,57){\vector(1,0){55}}
\put(30, 60){$^{()^{*}}$}
\put(65,52){\vector(-1,0){55}}
\put(30,40){$^{()_{*}}$}
\put(0,0){R(A)}
\put(2,45){\vector(0,-1){35}}
\put(10,25){$i$}
\put(-5,25){$\rho$}
\put(7,10){\vector(0,1){35}}
\put(18,10){\vector(2,1){67}}
\put(35,30){$\Psi$}
\put(40,15){$\Phi$}
\put(75,45){\vector(-2,-1){60}}
\end{picture}
\end{center}  
\end{remark}
\begin{remark}
Consider the function $\mu_{A} : L(A)\to R(A)$ defined by $\mu_{A}(\l(c))=\rho(c)$ , for any $c\in K(A)$ (it is easy to see that $\mu_{A}$ is well defined). Thus $\mu_{A}$ is an injective morphism of bounded lattices and the following diagram is commutative.
\begin{center}
\begin{picture}(150,70)
\put(-5,45){$K(A)$}
\put(55,55){$\rho_{A}$}
\put(20,50){ \vector(1,0){80}}
\put(105,45){ $R(A)$}
\put(10,40){\vector(2,-1){45}}
\put(20,20){$\l_{A}$}
\put(65,20){\vector(2,1){45}}
\put(85,20){$\mu_{A}$}
\put(50,5){$L(A)$}
\end{picture}
\end{center}
\end{remark}

\begin{lemma}\label{compacts} Let $u: A\to B $ an unital morphism of quantales  that preserves the compacts. For all $c,d \in K(A)$ the following hold: 
\begin{list}{(\arabic{nr})}{\usecounter{nr}}
\item If $ \rho_{A} (c)=\rho_{A}(d)$ then $ \rho_{B}(u(c))=\rho_{B}(u(d))$;
\item If $\l_{A}(c)=\l_{A}(d)$ then $\l_{B}(u(c))=\l_{B}(u(d))$.
\end{list} 
\end{lemma}
\begin{proof}
$(1)$ Assume $\rho_{A}(c)=\rho_{A}(d)$, hence $c \leq \rho_{A}(d)$. According to Lemma  \ref{c^{k}},(2), there exists an integer $k\geq 1 $ such that $c^{k}\leq d $, hence $(u(c))^{k}=u(c^{k})\leq u(d) $. Applying again Lemma \ref{c^{k}},(2), it follows that $u(c)\leq \rho_{B}(u(d))$, so $\rho_{B}(u(c))\leq \rho_{B}(\rho_{B}(u(d)))=\rho_{B}(u(d)).$ The converse inequality $\rho_{B}(u(d))\leq \rho(u(c))$ follows similarly.   \\
$(2)$ By $(1)$.
\end{proof}
Let $ u:A\to B$ a unital morphism of quantales that preserves the compacts. According to lemma \ref{compacts}, one can define a function $L(u):L(A)\to L(B)$ by $L(u)(\l_{A}(c))=\l_{B}(u(c))$, for any $c\in K(A) $.
\begin{proposition}\label{KlLmorphism}
$L(u)$ is a morphism of bounded lattices and the following diagram is commutative.
\begin{center}
\begin{picture}(150,70)
\put(0,50){$K(A)$}
\put(25,55){\vector(1,0){100}}
\put(70,60){$u|_{K(A)}$}
\put(130,50){$K(B)$}
\put(5,45){\vector(0,-1){30}}
\put(-10,30){$\l_{A}$}
\put(0,0){$L(A)$}
\put(25,5){\vector(1,0){100}}
\put(70,10){$L(u)$}
\put(130,0){$L(B)$}
\put(135,45){\vector(0,-1){30}}
\put(140,30){$\l_{B}$}
\end{picture}
\end{center}
\end{proposition}
\section{The boolean center of a quantale versus reticulation}
 \hspace{0.5cm}Let $A$ be a quantale and $B(A)$ the set of complemented elements of $A$. It is well-known that $B(A)$ is a Boolean algebra \cite{Jipsen}, \cite{Kowalski} that generalizes the Boolean centers of many concrete structures: commutative rings, bounded distributive lattices, residuated lattices, congruence distributive algebras,etc. For this reason, $B(A)$ will be called the Boolean center of the quantale $A$. In this section we shall prove that $B(A)$ is isomorhic with other two Boolean algebras: $B(L(A))$, the Boolean center of the reticulation $L(A)$ and $B(R(A))$, the Boolean center of the frame $R(A)$. This result will be used in proving a characterization theorem for hyperarhimedean quantales.

Let $A$ be a quantale and $B(A)$ the set of complemented elements of $A$.
\begin{lemma}\cite{Jipsen},\cite{GLC}\label{propB(A)} For all $a,b\in A $ and $e\in B(A)$ the following hold:
\begin{list}{(\arabic{nr})}{\usecounter{nr}}
\item $a\in B(A)$ iff $a \lor a^{\bot}=1 $;
\item $a\land e=a\cdot e$;
\item $e\to a=e^{\bot}\lor a$;
\item If $a\lor b=1$ and $a\cdot b=0$ then $a, b \in B(A)$;
\item $(a\land b)\lor e=(a\lor e)\land (b\lor e)$
\end{list}
\end{lemma}
\begin{lemma}\label{unitalMorph}
Let $u:A\to A'$ be a unital morphism of quantales. Then the following hold:
\begin{list}{(\arabic{nr})}{\usecounter{nr}}
\item If $e\in B(A)$ then $u(e)\in B(A')$;
\item $u_{|_{B(A)}}:B(A)\to B(A')$ is a boolean morphism.
\end{list}
\end{lemma}
\begin{proof}
(1) Assume $e\in B(A)$. So $e\lor f=1$, $e\cdot f=0$, for some $f\in A$. Thus $u(e)\lor u(f)=1$ and $u(e)\cdot u(f)=0$. By Lemma \ref{propB(A)},(4), one gets that $e\in B(A')$.\\
(2) Follows by applying Lemma \ref{propB(A)}.   
\end{proof}
Throughout this section we shall assume that $A$ is a coherent quantale.
\begin{lemma}
$ B(A)\subseteq K(A)$.
\end{lemma}
\begin{proof}
Let $e\in B(A)$ and $(a_{t})_{t\in T}\subseteq A$ such that $e\leq \bigvee _{t\in T}a_{t}$. Thus $ \bigvee_{t\in T} a_{t}\lor e^{\perp}=1$, hence there exists a finite subset $T_{0} $ of $T$ such that $ \bigvee_{t\in T_{0}} a_{t}\lor e^{\perp}=1$. According to Lemma \ref{propB(A)},  $ e\to \bigvee_{t\in T_{0}} a_{t}=1$, hence $e\leq \bigvee_{t\in T_{0}} a_{t}$. It follows that $ e\in K(A)$. 
\end{proof}
\begin{lemma}
If $e\in B(A) $ then $\l_{A}(e)\in B(L(A))$ and $\rho_{A}(e)\in B(R(A))$.
\end{lemma}
\begin{proof}
By using Lemma \ref{l(compact)} $ \l_{A}(e)\land \l_{A}(e^{\perp})=\l_{A}(e\cdot e^{\perp})=\l(0)=0$ and  $\l_{A}(e)\lor \l_{A}(e^{\perp})=\l(e\lor e^{\perp})=\l_{A}(1)=1$. Therefore $\l_{A}(e)\in B(L(A))$. Similarly, by using Lemma \ref{rho}, we can prove that $\rho_{A}(e)\in B(R(A))$.
\end{proof}
 According to previous Lemma, one can consider the functions:\\
$\l_{A}|_{B(A)}:B(A)\to B(L(A))$ and $\rho_{A}|_{B(A)}:B(A)\to B(R(A))$. 

Let us denote $B(\l_{A})=\l_{A}|_{B(A)}$ and $B(\rho_{A})=\rho_{A}|_{B(A)}$.
\begin{proposition}\label{l,rho,isomorphism}
$B(\l_{A})$ and $B(\rho_{A})$ are boolean isomorphisms.
\end{proposition} 
\begin{proof}
It is clear that $B(\l_{A})$ and $B(\rho_{A})$ are boolean morphisms. Thus their injectivity follows according to Lemma \ref{l(compact)},(3) and Lemma \ref{rho},(3). We shall prove that $B(\l_{A})$ is surjective. Let $\l(c)\in B(L(A))$, with $c\in K(A)$. Thus there exists $d\in K(A)$ such that $\l_{A}(c)\lor \l_{A}(d)=1$ and $\l_{A}(c)\land \l_{A}(d)=0$, hence by Lemma \ref{l(compact)}, (1) and (2), $\l_{A}(c\lor d)=1$ and $\l_{A}(c\cdot d)=0$. According to Lemma \ref{l(compact)},(3), is obtained $c\lor d=1$, therefore, by Lemma \ref{propQ},(3), $c^{n}\lor d^{n}=1$ for all integer $n \geq 1$. From $\l_{A}(c\cdot d)=0=\l(0)$ it results that $\rho_{A}(c\cdot d)=\rho_{A}(0)$, hence, by Lemma \ref{l(compact)},(4), there exists an integer $k\geq 1$ such that $c^{k}\cdot d^{k}=0$. According to Lemma \ref{propB(A)},(4), from $c^{k}\lor d^{k}=1$ and $c^{k}\cdot d^{k}=0$ it follows that $c^{k}, d^{k}\in B(A)$. By Lemma \ref{l(compact)},(5), $\l_{A}(c^{k})=\l(c)$, so $B(\l_{A})$ is surjective. Similarly, 
it follows that $B(\rho_{A})$ is surjective.
\end{proof}
If $\l_{A}(c)\in B(L(A))$ then $\mu_{A}(\l_{A}(c))=\rho_{A}(c)\in B(R(A))$. Hence one can consider the function 
$B(\mu_{A}):B(L(A))\to B(R(A))$ defined as $B(\mu_{A})=\mu_{A}|_{B(L(A))}$.
\begin{proposition}
$B(\mu_{A})$ is a boolean isomorphism and the following diagram is commutative:
\begin{center}
\begin{picture}(120,70)
\put(-5,45){$B(A)$}
\put(45,55){$B(\rho_{A})$}
\put(20,50){\vector(1,0){80}}
\put(105,45){$B(R(A))$}
\put(10,40){\vector(2,-1){45}}
\put(0,25){$B(\l_{A})$}
\put(60,20){\vector(2,1){45}}
\put(85,25){$B(\mu_{A})$}
\put(45,5){$B(L(A))$}
\end{picture}
\end{center}
\end{proposition}
A quantale $A$ is said to be {\em hyperarhimedean} if for any $c\in K(A)$ there exists an integer $n\geq 1$ such that $c^{n}\in B(A)$. Recall from \cite{Banaschewski} that a frame $L$ is zero-dimensional if any $a\in L$ is a join of complemented elements.
\begin{theorem}
\label{hyperarhimedian}
For a coherent quantale $A$, the following are equivalent:
\begin{list}{(\arabic{nr})}{\usecounter{nr}}
\item $A$ is hyperarhimedean.
\item $L(A)$ is a Boolean algebra.
\item $Max(A)=Spec(A)$.  \\
Moreover, if $A$ is semiprime then (1)-(3) are equivalent with the following assertion 
\item $R(A)$ is a zero-dimensional frame.
\end{list}
\end{theorem} 
\begin{proof}
(1) $\iff$ (2) According to Proposition \ref{l,rho,isomorphism}, the following properties are equivalent:
$L(A)$ is a Boolean algebra iff $L(A)=B(L(A))$ iff for any $c\in K(A)$, $\l(c)\in B(L(A))$ iff for any $c\in K(A)$, there exists $e\in B(A)$ such that $\l(c)=\l(e)$ iff for any $c\in K(A)$, there exists $e\in B(A)$ such that $\rho(c)=\rho(e)$ iff for any $c\in K(A)$, there exists $e\in B(A)$ and $n\geq 1$ such that $c^{n}\leq e\leq c$ iff for any $c\in K(A)$ there exists $n\geq 1$ such that $c^{n}\in B(A)$ iff $A$ is hyperarhimedian.\\
(1) $\iff$ (3) By Proposition \ref{l,rho,isomorphism} and the Nachbin Theorem \cite{BalbesDwinger}\\
(1) $\iff$ (4) Let $a$ be an arbitrary element in $A$, so $a=\lor _{i\in I}c_{i}$, with $c_{i}\in K(A)$ for any $i\in I$. Then $\rho(a)=\dot{\lor}_{i\in I}\rho(c_{i})$. According to the hypothesis (1), for each $i\in I$ there exists an integer $n_{i}\geq 1$ such that $c_{i}^{n_{i}}\in B(A_{i})$. Then $\rho(a)=\dot{\lor}_{i\in I}\rho(c_{i}^{n_{i}})$, with $\rho(c_{i}^{n_{i}})\in B(R(A))$ for each $i\in I$, so $R(A)$ is a zero-dimensional frame.\\
(4) $\iff$ (1) Let $c\in K(A)$ be such that $\rho(c)\in B(R(A))$. By Proposition \ref{l,rho,isomorphism}, there exists $e\in B(A)$, such that $\rho(c)=\rho(e)$, hence $c^{n}\leq e\leq c$ for some integer $n\geq 1$, therefore $c^{n}=e$. It follows that $(c\cdot e^{\perp})^{n}=0$. Since $A$ is semiprime it follows that $c\cdot e^{\perp}=0$, hence $c\leq e$, so $c=e\in B(A)$.
Let $c$ be an arbitrary element of $K(A)$. Then $\rho(c)\in K(R(A))$, hence, according to the hypothesis (4) and Proposition \ref{l,rho,isomorphism}, there exists a family $(e_{i})_{i\in I}$ in $B(A)$ such that $\rho(c)=\dot{\lor}_{i\in I}\rho(e_{i})$. Since $\rho(c)\in K(R(A))$, there exists a finite subset $I_{0}$ of $I$ such that $\rho(c)=\dot{\lor}_{i\in I}\rho(e_{i})\in B(R(A))$. By the previous remark, $c\in B(A)$.    
\end{proof}
\begin{remark}

According to Theorem \ref{hyperarhimedian} and the Nachbin theorem it follows that the hyperarhimedian bounded distributive lattices are exactly the Boolean algebras.
\end{remark}
\section{A lifting property}
 \hspace{0.5cm}Let $A$ be a commutative ring, $Idp(A)$ the Boolean algebra of its idempotents, $R(Id(A))$ the frame of radical ideals in $A$ and $B(R(Id(A)))$ the Boolean algebra of complemented elements of $R(Id(A))$. The two Boolean algebras $Idp(A)$ and $B(R(Id(A)))$ are isomorphic and the condition LIP can be expressed in terms of the frame $R(Id(A))$ (see \cite{Banaschewski}). Similar observations can be made in the case of the lifting properties for other concrete structures \cite{h}, \cite{e}, \cite{GeorgescuMuresan}, \cite{GeorgescuMuresanJMVL}, \cite{g}. Thus the lifting property (LP) in a coherent quantale $A$ will use the Boolean center $B(A)$ of $A$. The main tool for the study of this LP will be the Boolean isomorphisms established in Proposition \ref{l,rho,isomorphism}.
 
We fix a coherent quantale $A$. For any $a\in A$, consider the interval $[a)=\{x\in A|a\leq x\}$ and for all $x,y\in [a)_{A}$, denote $x\cdot_{a}y=x\cdot y\lor a$. It is easy to see that $[a)_{A}$ is closed under the new multiplication $\cdot_{a}$.
\begin{lemma}\label{[a)}
\begin{list}{(\arabic{nr})}{\usecounter{nr}}
\item $([a)_{A}, \lor,\land,\cdot_{a}, a, 1)$ is a coherent quantale such that $K([a)_{A})=K(A)\cap [a)_{A}$.
\item $\rho_{[a)}(x)=\rho_{A}(x)$, for any $x\in [a)_{A}$.
\end{list}
\end{lemma}
We denote $\rho=\rho_{A}$ and $\rho_{a}=\rho_{[a)_{A}}$. Let us consider the function $u_{a}^{A}:A\to [a)_{A}$ defined by $u_{a}^{A}(x)=x\lor a$, for any $x\in A$.
\begin{lemma}\label{uaA}
\begin{list}{(\arabic{nr})}{\usecounter{nr}}
\item $u_{a}^{A}$ is a unital quantale morphism.
\item If $c\in K(A)$ then $u_{a}^{A}(c)\in K([a)_{A})$.
\end{list}
\end{lemma}
\begin{proof}
(1) For all $x,y\in A$ the following hold: $u_{a}^{A}(x)\cdot_{a}u_{a}^{A}(y)=(x\lor a)\cdot (y\lor a)\lor a=xy \lor ax \lor ay\lor a^{2}\lor a=xy\lor a= u_{a}^{A}(xy)$. It is easy to see that $u_{a}^{A}$ preserves the arbitrary joins, therefore $u_{a}^{A}$ is a unital quantale morphism.\\
(2) Immediatelly. 
\end{proof}
\begin{lemma}\label{u comut rho}
The following diagram is commutative:
\begin{center}
\begin{picture}(150,70)
\put(70,60){$\rho$}
\put(0,50){$A$}
\put(20,55){\vector(1,0){105}}
\put(130,50){$R(A)$}
\put(5,45){\vector(0,-1){30}}
\put(0,0){$[a)_{A}$}
\put(-10,30){$u_{a}^{A}$}
\put(25,5){\vector(1,0){100}}
\put(70,10){$\rho_{a}$}
\put(130,0){$[\rho(a))_{R(A)}$}
\put(135,45){\vector(0,-1){30}}
\put(140,30){$u_{\rho(a)}^{R(A)}$}
\end{picture}
\end{center}
\end{lemma}
\begin{proof}
Consider an arbitrary element $x\in A$. Then, by Lemma \ref{[a)} (3), the following equalities hold:
$u_{\rho(a)}^{R(A)}(\rho(x))=\rho(a)\dot{\lor}\rho(x)=\rho(a\lor x)=\rho_{a}(a\lor x)=\rho_{a}(u_{a}^{A}(x))$.  
\end{proof}
We shall denote $\l=\l_{A}$ and $\l_{a}=\l_{[a)_{A}}$.
\begin{remark} According to Lemma \ref{uaA}(2), the quantale morphism $u_{a}^{A}$ preserves the compacts, so applying Proposition \ref{KlLmorphism}, the following diagram is commutative.
\begin{center}
\begin{picture}(150,70)
\put(70,60){$u_{a}^{A}$}
\put(0,50){$K(A)$}
\put(25,55){\vector(1,0){100}}
\put(130,50){$K([a)_{A})$}
\put(5,45){\vector(0,-1){30}}
\put(0,0){$L(A)$}
\put(-10,30){$\l$}
\put(25,5){\vector(1,0){100}}
\put(70,10){$L(u_{a}^{A})$}
\put(130,0){$L([a)_{A})$}
\put(135,45){\vector(0,-1){30}}
\put(140,30){$\l_{a}$}
\end{picture}
\end{center}
\end{remark}
\begin{proposition}\label{L([a))iso L(A)/a^{*}}
For any $a\in A$, the bounded distributive lattices $L([a)_{A})$ and $L(A)/a^{*}$ are isomorphic.
\end{proposition}
\begin{proof}
Let $x\in K(A)$. According to the diagram from previous remark, $L(u_{a}^{A})(\l(x))=\l_{a}(u_{a}^{A}(x))=\l_{a}(a\lor x)$. Therefore, by Lemma   \ref{l(compact)}(1), the following equivalences hold:
$L(u_{a}^{A})(\l(x))=0$ iff $\l_{a}(a\lor x)=\l_{a}(a)$ iff $\rho(a\lor x)=\rho_{a}(a)$ iff $\rho(a\lor x)=\rho(a)$ iff $\l(a\lor x)=\l(a)$ iff $\l(a)\lor \l(x)=\l(a)$ iff $\l(x)\leq \l(a)$ iff $\l(x)\in (\l(a)]_{A}=a^{*}$. Thus $Ker(L(u_{a}^{A}))=\{\l(x)|x\in K(A), L(u_{a}^{A})(\l(x))=0\}= a^{*}$, hence because $L(u_{a}^{A})$ is surjective, it follows that $L([a)_{A})$ and $L(A)/a^{*}$ are isomorphic.
\end{proof}
By Lemma \ref{unitalMorph}, we can consider the boolean morphism $B(u_{a}^{A})=u_{a}^{A}|_{_{B(A)}}: B(A)\to B([a)_{A})$. The following diagram is commutative: 
\begin{center}
\begin{picture}(150,70)
\put(0,50){$B(A)$}
\put(25,55){\vector(1,0){100}}
\put(70,60){$B(u_{a}^{A})$}
\put(130,50){$B([a)_{A})$}
\put(3,44){\oval(4,4)[t]}
\put(5,45){\vector(0,-1){30}}
\put(0,0){$K(A)$}
\put(25,5){\vector(1,0){100}}
\put(70,10){$u_{a}^{A}|_{K(A)}$}
\put(130,0){$K([a)_{A})$}
\put(133,44){\oval(4,4)[t]}
\put(135,45){\vector(0,-1){30}}
\end{picture}
\end{center}
The following definition is inspired by the condition (3) of Lemma 4 in \cite{Banaschewski}.
\begin{definition}
\item An element $a\in A$ has the lifting property (LP) if the boolean morphism $B(u_{a}^{A})$ is surjective.
\item The quantale  $A$ has LP if every element $a\in A$ has LP.
\end{definition}
The lifting property introduced by previous definition generalizes the condition LIP from ring theory \cite{a}, as well as the other boolean lifting properties existing in literature \cite{Banaschewski}, \cite{h}, \cite{e}, \cite{GCM},\cite{GLC},\cite{GeorgescuMuresan},\cite{GeorgescuMuresanJMVL},\cite{f}.
\begin{remark}
\begin{list}{(\arabic{nr})}{\usecounter{nr}}
\item  Any element a with the property $B([a)_{A})=\{0,1\}$ has LP.
\item If $p$ is an m-prime element of $A$ then one can prove that $B([p)_{A})=\{0,1\}$, therefore, by (1), $p$ has LP. Particularly, any maximal element of $A$ has LP.
\end{list}
\end{remark}  
\begin{theorem}
\label{LPequiv}
The following assertions are equivalent:
\begin{list}{(\arabic{nr})}{\usecounter{nr}}
\item The quantale $A$ has LP;
\item The frame $R(A)$ has LP;
\item The lattice $L(A)$ has Id-BLP;
\end{list}
\end{theorem}
\begin{proof}
(1) $\iff$ (2) According to Lemma\ref{u comut rho}, for any $a\in A$ the following diagram is commutative:
\begin{center}
\begin{picture}(150,75)
\put(70,61){$B(\rho)$}
\put(70,55){$\sim$}
\put(0,50){$B(A)$}
\put(25,55){\vector(1,0){100}}
\put(130,50){$B(R(A))$}
\put(5,45){\vector(0,-1){30}}
\put(-10,0){$B([a)_{A})$}
\put(-25,30){$B(u_{a}^{A})$}
\put(25,5){\vector(1,0){100}}
\put(70,5){$\sim$}
\put(70,11){$B(\rho_{a})$}
\put(130,0){$B([\rho_{a})_{R(A)})$}
\put(135,45){\vector(0,-1){30}}
\put(140,30){$B(u_{\rho(a)}^{R(A)})$}
\end{picture}
\end{center}
According to Proposition \ref{l,rho,isomorphism}, $B(\rho)$ and $B(\rho_{a})$ are boolean isomorphism, hence $B(u_{a}^{A})$ is surjective iff $B(u_{\rho(a)}^{R(A)})$ is surjective.\\
(1) $\imp$ (3) Let $I\in Id(L(A))$. According to the hypothesis (1), $B(u_{I_{*}}):B(A)\to B([I_{*}))$ is a surjective boolean morphism. According to the Lemma \ref{_^*},(2) $(I_{*})^{*}=I$ so by Proposition \ref{L([a))iso L(A)/a^{*}} we obtain a lattice isomorphism between $L([I_{*}))$ and $L(A)/I$. Consider the commutative diagrame: 
\begin{center}
\begin{picture}(250,70)
\put(70,60){$\l_{A}$}
\put(25,55){\vector(1,0){100}}
\put(0,50){$K(A)$}
\put(130,50){$L(A)$}
\put(5,45){\vector(0,-1){30}}
\put(-10,30){$u_{I_{*}}$}
\put(-5,0){$K([I_{*})_{A})$}
\put(70,10){$\l_{[I_{*})_{A}}$}
\put(35,5){\vector(1,0){90}}
\put(125,0){$L([I_{*})_{A})$}
\put(135,45){\vector(0,-1){30}}
\put(140,30){$L(u_{I_{*}})$}
\put(155,50){\vector(2,-1){80}}
\put(200,30){$p_{I}$}
\put(165,5){\vector(1,0){70}}
\put(190,5){$\sim $}
\put(240,0){$L(A)/I$}
\end{picture}
\end{center}
where $p_{I}$ is the lattice morphism associated to the ideal $I$.

It follows that in the categorie of Boolean algebras the following diagram is commutative:
\begin{center}
\begin{picture}(250,70)
\put(70,62){$B(\l_{A})$}
\put(70,55){$\sim$}
\put(25,55){\vector(1,0){100}}
\put(0,50){$B(A)$}
\put(130,50){$B(L(A))$}
\put(5,45){\vector(0,-1){30}}
\put(-25,30){$B(u_{I_{*}})$}
\put(-5,0){$B([I_{*})_{A})$}
\put(70,12){$B(\l_{[I_{*})_{A}})$}
\put(70,5){$\sim$}
\put(35,5){\vector(1,0){90}}
\put(125,0){$B(L([I_{*})_{A}))$}
\put(135,45){\vector(0,-1){30}}
\put(135,30){$B(L(u_{I_{*}}))$}
\put(165,45){\vector(2,-1){75}}
\put(200,30){$B(p_{I})$}
\put(180,5){\vector(1,0){60}}
\put(190,5){$\sim $}
\put(240,0){$B(L(A)/I)$}
\end{picture}
\end{center}
According to Proposition \ref{l,rho,isomorphism} $B(\l_{A})$ and $B(\l_{I_{*}})$ are boolean isomorphisms, hence, by the previous diagram the following implication holds: $B(u_{I_{*}})$ is surjective $\imp$ $B(p_{I})$ is surjective. We have proven that $B(p_{I})$ is surjective for every ideal $I$ of $L(A)$, therefore the lattice $L(A)$ has Id-BLP.\\
(3)$\imp$ (1) Let $a\in A$ so $a^{*}$ is an ideal of $L(A)$. By the hypothesis (3), the boolean morphism $B(p_{a^{*}}):B(L(A))\to B(L(A)/a^{*})$ is surjective. Let us consider the following commutative diagram:
\begin{center}
\begin{picture}(250,70)
\put(70,60){$\l_{A}$}
\put(25,55){\vector(1,0){100}}
\put(0,50){$K(A)$}
\put(130,50){$L(A)$}
\put(5,45){\vector(0,-1){30}}
\put(-10,30){$u_{a}$}
\put(-5,0){$K([a)_{A})$}
\put(70,10){$\l_{[a)_{A}}$}
\put(30,5){\vector(1,0){90}}
\put(125,0){$L([a)_{A})$}
\put(135,45){\vector(0,-1){30}}
\put(140,30){$L(u_{a})$}
\put(155,50){\vector(2,-1){80}}
\put(205,30){$p_{a^{*}}$}
\put(165,5){\vector(1,0){70}}
\put(190,5){$\sim $}
\put(240,0){$L(A)/a^{*}$}
\end{picture}
\end{center}
where $p_{a^{*}}:L(A)\to L(A)/a^{*}$ is the lattice morphism associated to the ideal $a^{*}$ and the lattice isomorphism between $L([a)_{A})$ and $L(A)/a^{*}$ is due to Proposition \ref{L([a))iso L(A)/a^{*}}. Thus is obtained the following commutative diagram in the category of Boolean algebras:
\begin{center}
\begin{picture}(250,70)
\put(70,60){$B(\l_{A})$}
\put(25,55){\vector(1,0){100}}
\put(0,50){$B(A)$}
\put(130,50){$B(L(A))$}
\put(5,45){\vector(0,-1){30}}
\put(-25,30){$B(u_{a})$}
\put(-5,0){$B([a)_{A})$}
\put(70,10){$B(\l_{a})$}
\put(30,5){\vector(1,0){90}}
\put(120,0){$B(L([a)_{A}))$}
\put(135,45){\vector(0,-1){30}}
\put(135,30){$B(L(u_{a}))$}
\put(160,45){\vector(2,-1){75}}
\put(205,30){$B(p_{a^{*}})$}
\put(170,5){\vector(1,0){65}}
\put(190,5){$\sim $}
\put(235,0){$B(L(A)/a^{*})$}
\end{picture}
\end{center}
By the hypothesis (3), $B(p_{a^{*}})$ is surjective , therefore, from previous commutative diagram, we get that $B(u_{a})$ is surjective. Conclude that the quantale $A$ has LP.       
\end{proof}
A quantale $A$ is local iff $|Max(A)|=1$; $A$ is semilocal iff has a finite number of maximal elements.
A bounded distributive lattice $L$ is said to be Id-local iff $|Max_{Id}(L)|=1$; A is Id-semilocal iff $Max_{Id}(L)$ is a finite set.
\begin{corollary}
The following assertions are equivalents:
\begin{list}{(\arabic{nr})}{\usecounter{nr}}
\item $A$ is a local quantale;
\item $R(A)$ is a local frame:
\item $L(A)$ is Id-local.
\end{list}
\end{corollary}
\begin{corollary}
 The following assertions are equivalent:
\begin{list}{(\arabic{nr})}{\usecounter{nr}}
\item $A$ is a semilocal quantale;
\item $R(A)$ is a semilocal frame:
\item $L(A)$ is Id-semilocal.
\end{list}
\end{corollary}
\begin{corollary}
Any local quantale has LP
\end{corollary}
\begin{proof}
 If $A$ is local then $L(A)$ is an Id-local bounded distributive lattice.According to Proposition 18, \cite{h}, the bounded distributive lattice $L(A)$ has Id-BLP. Thus, by Proposition \ref{LPequiv}, the quantale $A$ has LP.
\end{proof}
\begin{corollary}
If the reticulation $L(A)$ is a chain then the quantale $A$ has LP.
\end{corollary}
\begin{proof}
By Corollary 4, \cite{h} , the chain $L(A)$ has Id-BLP, hence, by Proposition \ref{LPequiv}, it follows that $A$ has BLP.
\end{proof}
\begin{corollary}
Any hyperarhimedian quantale $A$ has LP.
\end{corollary}
\begin{proof}
By Proposition \ref{hyperarhimedian}, $L(A)$ is a Boolean algebra, hence it has Id-BLP. Then we apply Proposition \ref{LPequiv}.  
\end{proof}
\begin{proposition}\label{A_lp=>[a)_lp}
Let $a\in A$. If $A$ has LP then $[a)_{A}$ has LP.
\end{proposition}
\begin{proof}
Let $b\in [a)_{A}$.We remark that $[b)_{[a)_{A}}=[b)_{A}$ and the following diagram is commutative
\begin{center}
\begin{picture}(100,70)
\put(0,50){$A$}
\put(75,50){$[a)_{A}$}
\put(75,0){$[b)_{A}$}
\put(15,55){\vector(1,0){60}}
\put(10,45){\vector(2,-1){60}}
\put(80,45){\vector(0,-1){30}}
\put(40,60){$u_{a}^{A}$}
\put(20,20){$u_{b}^{A}$}
\put(85,30){$u_{b}^{[a)_{A}}$}
\end{picture}
\end{center}
therefore the following diagram of Boolean algebras is commutative
\begin{center}
\begin{picture}(100,70)
\put(-15,50){$B(A)$}
\put(75,50){$B([a)_{A})$}
\put(75,0){$B([b)_{A})$}
\put(15,55){\vector(1,0){60}}
\put(10,45){\vector(2,-1){60}}
\put(80,45){\vector(0,-1){30}}
\put(35,60){$B(u_{a}^{A})$}
\put(20,20){$B(u_{b}^{A})$}
\put(85,30){$B(u_{b}^{[a)_{A}})$}
\end{picture}
\end{center}
Since $A$ has LP, the boolean mophisms $B(u_{a}^{A})$ and $B(u_{b}^{A})$ are surjective. From the previous commutative diagram it follows that $B(u_{b}^{[a)_{A}})$ is also surjective , hence $[a)_{A}$ has LP.
\end{proof}
If $u:A\to B$ is a quantale morphism then we denote $Ker(u)=\bigvee\{a\in A|u(a)=0\}$.
\begin{lemma}\label{u inj iff Ker(u)=0}
Let $u:A\to B$ a unital quantale morphism. Then $u$ is injective iff $Ker(u)=0$ 
\end{lemma}
\begin{proof}
Assume that $Ker(u)=0$ and consider two elements $x,y\in A$, $u(x)\leq u(y)$. Since $u(y^{\perp})\leq u(y)^{\perp}$ the following implications hold:
$u(x)\leq u(y) \imp u(x\cdot y^{\perp})=u(x)\cdot u(y^{\perp})\leq u(x)\cdot u(y)^{\perp}\leq 0 \imp u(x\cdot y^{\perp})=0\imp x\cdot y^{\perp}=0 \imp x\leq y$. Thus u is injective. The converse implication is immediate.
\end{proof}
\begin{corollary}\label{u surj=>u iso pe Ker(u)}
If the morphism $u:A\to B$ is surjective then $u|_{[Ker(u))_{A}}:[Ker(u))_{A}\to B$ is a quantale isomorphism.
\end{corollary}
\begin{proof}
We remark that $Ker(u|_{[Ker(u))_{A}})=Ker(u)$ so by previous Lemma,  $u_{[Ker(u))_{A}}$ is injective. 
\end{proof}
\begin{corollary}
If the unital quantale morphism $u:A \to B$ is surjective and $A$ has LP, then $B$ has LP.
\end{corollary}
\begin{proof}
Since $u|_{[Ker(u))_{A}}:[Ker(u))_{A}\to B $ is a quantale isomorphism  and by Proposition \ref{A_lp=>[a)_lp}, $[Ker(u))_{A}$ has LP, it follows that $B$ has LP.   
\end{proof} 
\section{Normal and B-normal quantales}
 \hspace{0.5cm}In this section we shall establish some connections between LP and two important classes of quantales: normal and B-normal quantales.
Recall from \cite{Rosenthal} that a quantale $A$ is normal if for any $a,b\in A$ with $a\lor b=1$ there exist $e,f\in A$ such that $a\lor e=b\lor f=1$ and $e\cdot f=0$. If $A$ is a frame then is obtained the notion of normal frame \cite{Johnstone}. The main exemple of normal quantale is the set of ideals in a Gelfand ring \cite{Johnstone} and the main exemple of normal frame is the set of ideals in a normal lattice \cite{Johnstone}, \cite{Cornish}, \cite{Pawar}, \cite{GeorgescuVoiculescu}. It is well-known that a bounded distributive lattice $L$ is normal iff $Id(L)$ is a normal frame.\\
Let us fix a coherent quantale $A$.
\begin{lemma}\label{normalIffK}
The following are equivalent:
\begin{list}{(\arabic{nr})}{\usecounter{nr}}
\item $A$ is normal;
\item For all $c,d\in K(A)$, with $c\lor d=1$ there exist $e,f\in K(A)$ such that $c\lor e=d\lor f=1$ and $e\cdot f =0$.
\end{list}
\end{lemma}  
\begin{proof}
(1) $\imp$ (2) Let $c,d \in K(A)$ such that $c\lor d=1$. By (1) there exist $c_{1}, d_{1}\in A , c\lor c_{1}=d\lor d_{1}=1$ and $c_{1}\cdot d_{1}=0$. Since $1$ is compact there exist $e,f\in K(A) , e\leq c_{1}, f\leq d_{1}$, such that $c\lor e=d\lor f=1$ and obviously $e\cdot f =0$.\\
(2)$\imp$ (1) Let $a,b\in A$ with $a\lor b=1$. Since $1\in K(A)$, it follows the existence of two elements $c,d \in K(A)$ such that $c\leq a , d\leq b$ and $c\lor d=1$. By (2), $c\lor e=d\lor f=1$ and $e\cdot f =0$ for some elements $e,f\in K(A)$. It is obvious that $a\lor e= b\lor f=1$.
\end{proof}
\begin{proposition}
The following are equivalent:
\begin{list}{(\arabic{nr})}{\usecounter{nr}}
\item $A$ is a normal quantale;
\item $R(A)$ is a normal frame;
\item $L(A)$ is a normal lattice.
\end{list}
\end{proposition}
\begin{proof}
(1) $\imp$ (2) Let $a,b\in A$ such that $\rho(a)\dot \lor \rho(b)=1$, hence, by Lemma \ref{rho}(6), $a\lor b=1$. Thus, by hypothesis that $A$ is normal, there exist $c,d\in A$ such that $a\lor c =b\lor d=1 $ and $c\cdot d =0$. It follows that $\rho(a)\dot \lor \rho(c)=\rho(b) \dot \lor \rho (d) =1$ and $\rho(c)\land \rho(d)=\rho(c \cdot d)=\rho(0)$ , therefore $R(A)$ is a normal frame.\\
(2)$\imp$ (1) Let $a,b$ be two elements of $A$ such that $a\lor b=1$, hence $\rho(a)\dot \lor \rho(b)=1$. According to the normality of the frame $R(A)$,  there exist $c,d\in A$ such that $\rho(a)\dot\lor \rho(c)=\rho(b)\dot\lor \rho(d)$ and $\rho(c)\land \rho(d)=\rho(0)$. By using Lemma \ref{rho},(6) and (2), it follows that $a\lor c =b\lor d=1$ and $\rho(c\cdot d)=\rho(0).$ Thus by Lemma \ref{c^{k}},(3), there exists an integer $n\geq 1$ such that $c^{n}\cdot d^{n}=0$. Finally we have $a\lor c^{n}=b\lor d^{n}=1$ and $c^{n}\cdot d^{n}=0$, so $A$ is a normal quantale.\\
(2) $\iff$ (3) According to Corollary \ref{Phi,PsiIsomorphism}, the frames $R(A)$ and $Id(L(A))$ are isomorphic, hence the following equivalences hold: $R(A)$ is a normal frame iff $Id(L(A))$ is a normal frame iff $L(A)$ is a normal lattice.

Recall from \cite{Cignoli} that a bounded distributive lattice $L$ is $B-normal$ if for all $a,b\in L$ such that $a\lor b=1$ there exist $e,f\in B(L)$ such that $a\lor e=b\lor f=1$ and $e\land f=0$.
\end{proof}
\begin{definition}
A quantale $A$ is B-normal if for any $a,b\in A$ such that $a\lor b=1$ there exist $e,f\in B(A)$ such that $a\lor e=b\lor f =1$ and $e\cdot f=0$.

 If the B-normal  quantale $A$ is a frame then say that $A$ is a B-normal frame. 
\end{definition}
We remark that weakly zero-dimensional frames of \cite{Banaschewski} coincide to the B-normal frames.

Of course a B-normal quantale is normal. Similar to Lemma \ref{normalIffK} one can prove the following result.
\begin{lemma}
The following assertions are equivalent:
\begin{list}{(\arabic{nr})}{\usecounter{nr}}
\item $A$ is a B-normal quantale;
\item For all $c,d\in K(A)$ with $c\lor d=1$ there exist $e,f\in B(A)$ such that $c\lor d=e\lor f=1$ and $e\cdot f=0$.
\end{list}
\end{lemma}
One can prove that a bounded distributive lattice $L$ is B-normal  iff the frame $Id(L)$ is B-normal.
\begin{proposition}\label{BnormalEquiv}
The following assertions are equivalent:
\begin{list}{(\arabic{nr})}{\usecounter{nr}}
\item $A$ is a B-normal quantale;
\item $R(A)$ is a B-normal frame;
\item $L(A)$ is a B-normal lattice.
\end{list}\end{proposition}
\begin{proof}
(1) $\imp$ (2) Let $a,b\in A$ with $\rho(a)\dot \lor \rho(b)=1$, so $a\lor b=1$. By the hypothesis, there exist $e,f\in B(A)$ such that $a\lor e=b\lor f=1$ and $e\cdot f=0$, therefore $\rho(a)\dot\lor \rho(e)=\rho(b)\dot\lor\rho(f)=1$ and $\rho(e)\land \rho(f)=\rho(e\cdot f)=\rho(0)$. Since $\rho(e), \rho(f)\in B(R(A))$, it follows that $R(A)$ is a B-normal frame.\\
(2) $\imp$ (1) Let $a,b\in A$ with $a\lor b=1$, hence $\rho(a)\dot\lor \rho(b)=1$. Thus there exist $e_{1}, f_{1}\in B(R(A))$ such that $\rho(a)\dot\lor e_{1}=\rho(b)\dot\lor f_{1}=1$ and $e_{1}\land f_{1}=\rho(0)$. According to Proposition \ref{l,rho,isomorphism} there exist $e,f\in B(A)$ such that $e_{1}=\rho(e)$ and $f_{1}=\rho(f)$, hence $\rho(a)\dot\lor \rho(e)=\rho(b)\dot \lor \rho(f)=1$ and $\rho(e\cdot f)=e_{1}\land f_{1}=\rho(0)$. These properties imply $a\lor e=b\lor f =1$ and $e\cdot f =0$, hence $A$ is B-normal. \\
(2) $\iff$ (3) The frames $R(A)$ and $Id(L(A))$ are isomorphic (by Corollary \ref{Phi,PsiIsomorphism}), therefore $R(A)$ is a B-normal frame iff $L(A)$ is a B-normal lattice. 
\end{proof}
The following result connects the properties of Propositions \ref{LPequiv} and \ref{BnormalEquiv}.
\begin{theorem}\label{LP equiv Bnormal}
The following properties are equivalent:
\begin{list}{(\arabic{nr})}{\usecounter{nr}}
\item the quantale $A$ has LP;
\item the frame $R(A)$ has LP;
\item the lattice $L(A)$ has Id-BLP;
\item the quantale $A$ is B-normal;
\item the frame $R(A)$ is B-normal;
\item the lattice $L(A)$ is B-normal.
\end{list}
\end{theorem}
\begin{proof}
(1) $\iff$ (2) $\iff$ (3) By Proposition \ref{LPequiv}\\
(4) $\iff$ (5) $\iff$ (6) By Proposition \ref{BnormalEquiv}\\
(2) $\iff$ (5) By Lemma 4 of \cite{Banaschewski}
\end{proof}

By definition, the Jacobson radical of the quantale $A$ is $r(A)=\land Max(A)$. If $A$ is the quantale $Id(R)$ of the ideals of a commutative ring $R$ then $r(A)$ is exactly the Jacobson radical of $R$.
\begin{lemma}\label{a lor r(A)=1 then a=1}
Let $a$ be an element of the quantale $A$. If $a\lor r(A)=1$ then $a=1$.  
\end{lemma}
\begin{proof}
If $a\not= 1$ then $a\leq m$ for some $m\in Max(A)$. Then $m=a\lor m \geq a\lor r(A)=1$, contradicting that $m\in Max(A)$.
\end{proof}
\begin{proposition}
If $A$ is a normal quantale then $r(A)$ has LP.
\end{proposition}
\begin{proof}
We have to prove that the boolean morphism $B(u_{r(A)}):B(A)\to B([r(A))_{A})$ is surjective. Let $x\in B([r(A))_{A})$ so there exists $y\geq r(A)$ such that $x\lor y=1$ and $x\land y=r(A)$. Since $A$ is normal there exist $a,b\in A$ such that $x\lor a=y\lor b=1$ and $a\cdot b=0$. Then $a\lor b\lor x=a\lor b\lor y=1$, hence, by Lemma \ref{propQ},(2), $a\lor b\lor r(A)=a\lor b\lor (x\land y)=1$. According to Lemma \ref{a lor r(A)=1 then a=1}, $a\lor b=1$. Applying Lemma \ref{propB(A)}, from $a\lor b=1$ and $a\cdot b=0$ we obtain $a,b\in B(A)$. Due to Lemma \ref{propB(A)}: $B(u_{r(A)}(b))=b\lor r(A)=b\lor (x\land y)=(b\lor x)\land (b\lor y)=b\lor x$. From $x\lor a=1$ we infere that $b=b\cdot(x\lor a)=b\cdot x\lor b\cdot a=b\cdot x=b\land x$, hence $b\leq x$. Thus $b\lor x=x$, so $B(u_{r(A)})(b)=x$. Therefore $B(u_{r(A)})$ is surjective.  
\end{proof}
The notion of property (*) there was defined for residuated lattice in \cite{GeorgescuMuresan}, for bounded distributive lattice in \cite{h} and for congruence distributive universal algebras in \cite{GeorgescuMuresanJMVL}. The following definition proposes a notion of property (*) in the framework of the quantales.
\begin{definition}
A quantale $A$ has the property (*) if for any $a\in A$ there exist $c\in K(A)$, $c\leq r(A)$ and $e\in B(A)$ such that $a=c\lor e$.
\end{definition}
\begin{proposition}\label{* => LP}
If the quantale $A$ has the property (*) then $A$ has LP.
\end{proposition}
\begin{proof}
Let $a,b\in A$, such that $a\lor b=1$. From the property (*) follows the existence of $c,d\in K(A)$ and $e,f\in B(A)$ such that $c,d\leq r(A)$, $a=c\lor e$, $b=d\lor f$. Then $e\lor f\lor r(A)\geq c\lor d\lor e\lor f=a\lor b=1$, hence by Lemma \ref{a lor r(A)=1 then a=1}, it follows that $e\lor f=1$. We remark that $e^{\perp}\land f^{\perp}=(e\lor f)^{\perp}=0$, $a\lor e^{\perp}=c\lor e\lor e^{\perp}=1$ and $b\lor f^{\perp}=d\lor f\lor f^{\perp}=1$. Thus  $A$ is B-normal, hence, by Proposition  \ref{LP equiv Bnormal}, $A$ has LP.
\end{proof}
\begin{proposition}
If the quantale $A$ has the property (*) then the frame $R(A)$ has the property (*).
\end{proposition}
\begin{proof}
Assume that $A$ has the property (*). Let $x\in R(A)$ so $x=\rho(a)$ for some $a\in A$. Then there exists $c\in K(A)$ and $e\in B(A)$ such that $c\leq r(a)$ and $a=c\lor e$, hence $\rho(e)\in B(R(A))$, $\rho(c)\in K(R(A)$, $\rho(c)\leq \rho(r(A))\leq \rho(R(A))$ and $x=\rho(a)=\rho(c\lor e)=\rho(c)\dot{\lor} \rho(e)$. 
\end{proof}
\begin{proposition}
Let $a\in A $. If $A$ has the property $(*)$ then the quantale $[a)_{A}$ has the property (*).
\end{proposition}
\begin{proof}
Let $x\in[a)_{A}$. Since $A$ has the property (*) there exist $c\in K(A)$ and $e\in B(A)$ such that $c\leq r(A)$ and $x=c\lor e$. We remark that $x=(c\lor a)\lor(e\lor a)$, $c\lor a\in K([a)_{A})$ and $e\lor a\in B([a)_{A})$. We also have $c\lor a\leq r(A)\lor a=(\land Max(A))\lor a \leq \land \{m\in Max(A)|a\leq m\}=r([a)_{A})$. Therefore $[a)_{A}$ has the property (*).
\end{proof}
\begin{corollary}\label{u surj,A * =>B*}
If a quantale morphism $u:A\to B$ is surjective and $A$ has the property (*) then $B$ has the property (*).
\end{corollary}
\begin{proof}
By previous Proposition and Corollary \ref{u surj=>u iso pe Ker(u)}.
\end{proof}
\section{Finite products of quantales and LP}
 \hspace{0.5cm}In this section we shall study the condition LP in finite products of quantales.We shall prove that the semilocal quantales with LP are exactly finite products of local quantales.

\begin{lemma}
Let $A_{1},..,A_{n}$ be coherent quantales and their product $A=\prod _{i=1}^{n}A_{i}$. Then $K(A)=\prod_{i=1}^{n}K(A_{i})$ and $A$ is a coherent quantale.
\end{lemma}
\begin{proposition}\label{[a) iso prod[ai)}
Let $A$ be a quantale and $a, a_{1}, .., a_{n}\in A$ such that $a=\land _{i=1}^{n}a_{i}$ and $a_{i}\lor a_{j}=1$ for all $i\not= j$. Then  the function $u:[a)_{A}\to \prod_{i=1}^{n}[a_{i})_{A}$, defined by $u(x)=(x\lor a_{1}, ..., x\lor a_{n})$, for any $x\in [a)_{A}$, is a quantale isomorphism. 
\end{proposition}
\begin{proof}
To prove that $u$ is a quantale morphism is straightforward. For example, if $x,y\in [a)$ then for any $i\in\{1,..,n\}$, the following equalities hold: $x\cdot y\lor a_{i}=(x\lor a_{i})\cdot (y\lor a_{i})\lor a_{i}=(x\lor a_{i})\cdot_{a_{i}}(y\lor a_{i})$, therefore $u(x\cdot y)=u(x)\cdot u(y)$. One remarks that 
$Ker(u)=\lor \{ x\geq a|x\lor a_{i}=a_{i},$ for all $i\in\{1,..,n\}\}
=\lor \{x\geq a |x\leq a_{i},$ for all $i\in\{1,..,n\} \}=
\lor\{x\geq a| x\leq \land_{i=1}^{n}a_{i}\}= a$. By Lemma \ref{u inj iff Ker(u)=0}, $u$ is injective.
Consider the elements $x_{i}\in [a_{i})_{A}, i\in\{1,..,n\}$, so $a_{i}\lor x_{j}=1$ for all $j\not= i$. According to Lemma \ref{propQ},(2), $a_{i}\lor (\prod _{j\not= i}x_{j})=1$, for $i\in\{1,..,n\}$. Since $a_{i}\leq x_{i}$, by using Lemma \ref{propQ},(4), it follows that $a_{i}\lor(\prod_{j\not= i}x_{j}\cdot x_{i})=x_{i}$ for all $i\in\{1,..,n\}$. Therefore $a_{i}\lor \prod_{j=1}^{n}x_{j}=x_{i}$, for $i\in\{1,..,n\}$, hence $u(\prod_{j=1}^{n}x_{j})=(x_{1}, .., x_{n})$, i.e. $u$ is surjective. 
\end{proof}
\begin{corollary}
If $\land_{i=1}^{n}a_{i}=0$ and $a_{i}\lor a_{j}=1$ for $i\not= j$ then the function $u:A\to \prod_{i=1}^{n}[a_{i})$ defined by $u(x)=(x\lor a_{1}, .., x\lor a_{n})$ for $x\in A$ is a quantale isomorphism. 
\end{corollary}
\begin{proposition}\label{(A iso prodAi)iff (Ai iso[ei))}
If $A, A_{1}, .., A_{n} $ are coherent quantales then the following are equivalent :
\begin{list}{(\arabic{nr})}{\usecounter{nr}}
\item $A$ and $\prod_{i=1}^{n}A_{i}$ are isomorphic ;
\item There exist $e_{1}, .., e_{n}\in B(A)$ such that $\land_{i=1}^{n}e_{i}=0$, $e_{i}\lor e_{j}=1$ for $i\not= j$ and $A_{i}, [e_{i})_{A}$ are isomorphic for any $i\in \{1,..,n\}$.
\end{list}
\end{proposition}
\begin{proof}
(1) $\imp$ (2) Assume $A=\prod_{i=1}^{n}A_{i}$ and denote $e_{i}=(1,..,1,0,1,..,1), i\in\{1,..,n\}$. If $\pi_{i}:A\to A_{i}, i\in\{1,..,n\}$ are the projections then we observe that $Ker(\pi_{i})=\lor\{  a\in A|\pi_{i}(a)=0\}=e_{i}$, for $i\in\{1,..,n\}$.
By Corollary \ref{u surj=>u iso pe Ker(u)}, one obtains the isomorphisms $A_{i}\simeq [Ker(\pi_{i}))_{A}=[e_{i})_{A} , i\in\{1,..,n\}$. The elements $e_{1}, .., e_{n}$ satisfy the condition $\land_{i=1}^{n}e_{i}=0$ and $e_{i}\lor e_{j}=1$, for $i\not= j$. \\
(2) $\imp$ (1) By previous Corollary, one gets the isomorphisms $A\simeq \prod_{i=1}^{n}[e_{i})_{A}\simeq \prod_{i=1}^{n}A_{i}$.  
\end{proof}
\begin{proposition}\label{MaxA}
If $A=\prod_{i=1}^{n}A_{i}$ then the following hold:
\begin{list}{(\arabic{nr})}{\usecounter{nr}}
\item $Max(A)=\cup_{i=1}^{n}(\{1\}\times ..\times Max(A_{i})\times ..\times\{1\})$;
\item $r(A)=(r(A_{1}), .., r(A_{n}))$;
\item $|Max(A)|=\sum _{i=1}^{n}|Max(A_{i})|$.
\end{list}
\end{proposition}
\begin{proof}
(1) Let $i\in \{1,.., n\}$ and $m_{i}\in Max(A)_{i}$. Denote $m=(1, .., 1, m_{i}, 1,.., 1)\in A$ and consider an element $a=(a_{1},.., a_{n})$ of $A$ such that $m\leq a < (1,.., 1)$. Then $m_{i}\leq a_{i}< 1$ and $a_{j}=1,$ for all $j\not= i$. Since $m_{i}\in Max(A_{i})$ it follows that $m_{i}=a_{i}$, hence $m=a$. Therefore $m=a$, so $\bigcup_{i=1}^{n}(\{1\}\times, .., \times Max(A_{i})\times, ..\{1\})\subseteq Max(A)$.\\
In order to prove the converse inclusion, let us consider $m=(m_{1},.., m_{n})\in Max(A)$, so $m_{i}< 1$ for some $i\in\{1,.., n\}$. Assume by absurd that there exists an index $j\not=i$, such that $m_{j}< 1$. Define $a=(a_{1},..a_{n})\in A$ by $a_{i}=m_{i}, a_{j}=m_{j}$ and $a_{k}=1$, for all $k\not=i,j$. Then $m\leq a< 1$, hence $m=a$. It follows that $m<(1, .., 1, a_{i}, 1, .., 1)<(1,.., 1)$, contradicting $m\in Max(A)$.
The assertions (2) and (3) follows from (1).
\end{proof}
\begin{proposition}\label{[r(A)) iso prod([mi))}
If $Max(A)=\{m_{1},..,m_{n}\}$ then $[r(A))_{A}\simeq \prod_{i=1}^{n}[m_{i})_{A}$.
\end{proposition}
\begin{proof}
We remark that $\land_{i=1}^{n}m_{i}=r(a)$ and $m_{i}\lor m_{j}=1$ for $i\not= j$, then we apply Proposition 
\ref{[a) iso prod[ai)}.
\end{proof}
\begin{lemma}\label{e complem, e<=r(A) => e=0}
If $e\in B(A)$ and $e\leq r(A)$ then $e=0$.
\end{lemma}
\begin{proof}
Assume $e^{\perp}\not= 1$ so there exists $m_{0}\in Max(A)$ such that $e^{\perp}\leq m_{0}$. Since $e\leq r(A)\leq m_{0}$ we have 
$1=e\lor e^{\perp}\leq m_{0}$. This contradiction shows that $e=0$.
\end{proof}
\begin{proposition}\label{LP si * (A iff Ai)}
Let the coherent quantales $A_{1},..,A_{n}$ and the product $A=\prod_{i=1}^{n}A_{i}$.
\begin{list}{(\arabic{nr})}{\usecounter{nr}}
\item $A$ has LP iff $A_{i}$ has LP, for $i\in\{1,..,n\}$;
\item $A$ has the property (*) iff $A_{i}$ has property (*), for $i\in\{1,..,n\}$. 
\end{list}
\end{proposition}
\begin{proof}
(1) Straightforward, by using Proposition \ref{LP equiv Bnormal};\\
(2) For direct implication we apply Corollary \ref{u surj,A * =>B*}. Conversely let $a=(a_{1}, .., a_{n})\in A$, so for any $i\in \{1,..n\}$, there exist $c_{i}\in K(A_{i}), e_{i}\in B(A)$ such that $c_{i}\leq r(A_{i})$, $a_{i}=c_{i}\lor e_{i}$. Denote $c=(c_{1}, .., c_{n})$ and $e=(e_{1}, .., e_{n})$. Thus $c\in K(A)$, $c\leq (r(A_{1}), .., r(A_{n}))=r(A)$ and $e\in B(A)$. It is obvious that $a=c\lor e$, hence $A$ has the property (*). 
\end{proof}
The following proposition characterize the finite product of local quantales by using the properties LP and (*). It generalizes Proposition 6.13 of \cite{GeorgescuMuresan}, Theorem 6.1 of \cite{GeorgescuMuresanJMVL}, as well as Theorem 2.10 of \cite{Lam}. 
\begin{theorem}
For a coherent quantale $A$ the following properties are equivalent:
\begin{list}{(\arabic{nr})}{\usecounter{nr}}
\item $A$ is a semilocal and has the property (*);
\item $A$ is semilocal and has LP;
\item $A$ is semilocal and $r(A)$ has LP;
\item There exist $e_{1}, .., e_{n}\in B(A)$ such that $\land _{i=1}^{n}e_{i}=0,$ $ e_{i}\lor e_{j}=1$ for $i\not= j$ and $[e_{i})_{A}$ is a local quantale, for any $i\in\{1,..,n\}$.
\item $A$ is isomorphic to a finite product of local quantales.
\end{list}
\end{theorem}
\begin{proof}
(1) $\imp$ (2) By Proposition \ref{* => LP}.\\
(2) $\imp$ (3) Obvious.\\
(3) $\imp$ (4) Assume  $Max(A)=\{m_{1}, .., m_{n}\}$, hence by Proposition \ref{[r(A)) iso prod([mi))}, the quantales $[r(A))_{A}$ and 
$\prod_{i=1}^{n}[m_{i})_{A}$ are isomorphic. Thus there exist $f_{1}, .., f_{n}\in B([r(A))_{A})$ such  that $r(A)=\land_{i=1}^{n}f_{i}$ and $f_{i}\lor f_{j}=1$ for $i\not= j $. Since the element $r(A)$ verifies LP there exist $e_{1}, .., e_{n}\in B(A)$ such that $f_{i}=e_{i}\lor r(A)$, for any $i\in\{1,..,n\}$. Applying Lemma \ref{propB(A)},(5), $(\land _{i=1}^{n}e_{i})\lor r(A)=\land_{i=1}^{n}(e_{i}\lor r(A))= \land_{i=1}^{n}f_{i}=r(A)$, hence $\land_{i=1}^{n}e_{i}\leq r(A)$. Since $\land _{i=1}^{n}e_{i}\in B(A)$, from Lemma \ref{e complem, e<=r(A) => e=0}, it follows that  $\land_{i=1}^{n}e_{i}=0$. For all $i\not= j$ we have $r(A)\lor e_{i}\lor e_{j}=f_{i}\lor f_{j}=1$ and $e_{i}\lor e_{j}\in B(A)$, therefore by Lemma \ref{a lor r(A)=1 then a=1}, we get $e_{i}\lor e_{j}=1$. Thus the quantales $A$ and $\prod_{i=1}^{n}[e_{i})_{A}$ are isomorphic, hence, by Proposition \ref{MaxA},(3), $|Max(A)|=\sum_{i=1}^{n}|Max[e_{i})_{A}|$. But $f_{i}\not= 1$ so $e_{i}\not= 1$ for each $i\in \{1,..,n\}$, hence $[e_{i})_{A}$ is non-trivial, for  $i=\overline {1,n}$. Thus $|Max[e_{i})_{A})|\geq 1$, for any $i\in\{1,..,n\}$, therefore $|Max([e_{i})_{A}|=1$, for $i\in\{1,..,n\}$, i.e. $[e_{1})_{A}, .., [e_{n})_{A}$ are local quantales.\\
(4) $\imp$ (5) By Proposition \ref{(A iso prodAi)iff (Ai iso[ei))}.\\
(5) $\imp$ (1) Any local quantale verifies the property (*), then we apply Proposition \ref{LP si * (A iff Ai)}
 \end{proof}
 \begin{corollary}
 If the quantale $A$ is semilocal, then: $A$ satisfies (*) iff A has LP iff $r(A)$ has LP.
\end{corollary}

\end{document}